\documentclass[11pt, onecolumn]{IEEEtran}
\IEEEoverridecommandlockouts

\usepackage{cite}
\usepackage{amsmath,amssymb,amsfonts,amsthm}
\usepackage{graphicx}
\usepackage{textcomp}
\usepackage{xcolor}
\usepackage{url} 

\usepackage{bbm}
\usepackage{mathtools}
\usepackage{subcaption}
\usepackage{multirow}
\usepackage{algorithm}
\usepackage[noend]{algpseudocode}

\DeclareMathOperator*{\argmin}{argmin}
\DeclareMathOperator{\EX}{\mathbb{E}}

\newtheorem{theorem}{Theorem}
\newtheorem{problem}{Problem}
\newtheorem{assumption}{Assumption}
\newtheorem{example}{Example}
\newtheorem{remark}{Remark}
\newtheorem{lemma}{Lemma}

\def\BibTeX{{\rm B\kern-.05em{\sc i\kern-.025em b}\kern-.08em
    T\kern-.1667em\lower.7ex\hbox{E}\kern-.125emX}}

\begin{document}

\title{RAIN-FIT: Learning of Fitting Surfaces and Noise Distribution from Large Data Sets}

\author{\IEEEauthorblockN{Omar M. Sleem}
\IEEEauthorblockA{\textit{Kyocera International, Inc.} \\
omar.sleem@kyocera.com}
\and
\IEEEauthorblockN{Sahand Kiani}
\IEEEauthorblockA{\textit{Pennsylvania State University}\\
szk6437@psu.edu}
\and
\IEEEauthorblockN{Constantino M. Lagoa}
\IEEEauthorblockA{\textit{Pennsylvania State University}\\
cml18@psu.edu}
\thanks{Research reported in this publication was supported by the National Heart, Lung, And Blood Institute of the National Institutes of Health under Award Number R61/33 HL164868, and the National Center for Advancing Translational Sciences, National Institutes of Health, through Grant UL1 TR002014 and Grant UL1 TR00045. The content is solely the responsibility of the authors and does not necessarily represent the official views of the NIH.}
}

\maketitle

\begin{abstract}
This paper presents a computationally efficient algorithm for estimating implicit surfaces from noisy point clouds. Modeling the surface as the zero level-set of a function spanned by known basis features, we estimate the surface representation with partial knowledge of the underlying noise distribution. The method operates without data preprocessing or hyperparameter tuning, features linear computational complexity with respect to sample size, and seamlessly generalizes to high-dimensional data. Furthermore, we provide theoretical guarantees for the almost sure convergence of the proposed algorithm. Numerical benchmarks against the state-of-the-art baselines, such as Poisson Reconstruction and Encoder-X, demonstrate superior performance in the presence of highly corrupted point clouds.
\end{abstract}


\section{Introduction}\label{sec:introduction}
Surface fitting is a fundamental mathematical technique used to model underlying data structures, essential for creating 3D models and providing ground-truth data for stereo reconstruction \cite{berger2014state}. The wide variety of available algorithms \cite{berger2013benchmark} differ primarily in their input requirements and output structures.

\subsection{Related Works}
Implicit Polynomials (IPs) are widely used for surface fitting due to their resolution-independent, compact parameterization that naturally encapsulates complex topologies. Classical IP fitting minimizes distances to zero-level sets but often suffers from coefficient sensitivity and oscillations \cite{1323797}. Alternatively, Poisson reconstruction provides smooth surfaces by solving a Poisson equation over the input data \cite{kazhdanpoisson}. However, it relies on a pre-computed vector field of oriented normals; even minor stochastic perturbations in the point cloud can misalign these normals and severely distort the reconstruction.

Recently, coordinate-based neural networks have established a new state-of-the-art by learning continuous spatial mappings \cite{tong2021polynomial}. Architectures like SIREN \cite{sitzmann2020implicit} and DeepSDF \cite{park2019deepsdf} capture fine geometric details, while methods like Encoder-X \cite{9336312} and DeepFit \cite{ben2020deepfit} target polynomial features and local structural weighting, respectively. Despite their efficacy, their strict reliance on local spatial coherence is their underlying vulnerability.

While these algorithms excel on pristine data, practical measurements are pervasively corrupted by noise. The existence of noisy data disrupts local geometry and normal vector calculations, causing methods reliant on local structure—such as Poisson reconstruction, SIREN, and Encoder-X—to fail as shown in \cite{sleem2026rain}. This exposes a critical need for a globally structured approach capable of ingesting corrupted point clouds without relying on fragile local geometric estimates. Furthermore, the proposed robust identification framework has the potential to be adapted for other diverse applications where measurement noise and data corruption are significant factors, such as the modeling of human physical activity \cite{li2025risk}.

\subsection{Key Contributions}
This paper introduces Robust Algorithm for Implicit Noisy Data Fitting (RAIN-FIT) to address estimating implicit surfaces from noisy point clouds. Building upon foundational concepts previously restricted to polynomial system identification in \cite{hojjatinia2020identification,kiani2026willems}, we significantly expand the framework to encompass a generalized, broad class of feature functions. By modeling the considered surface as the zero set of a function within the span of these feature functions, RAIN-FIT leverages partial knowledge of the noise distribution to simultaneously estimate the surface representation and the unknown noise parameters. This yields a compact, global description of the surface that naturally facilitates extrapolation tasks, such as the concise parameterization of point clouds.

Unlike traditional approaches, our globally structured approach bypasses the need for local geometric estimates, making it highly effective even with sparse samples and high-dimensional data. Another contribution of RAIN-FIT is its linear computational efficiency, which is achieved by decoupling the mathematical process into a feature-mapping computation and an accelerated surface estimation phase. The algorithm operates entirely without hyperparameter tuning or data preprocessing. Finally, we provide rigorous theoretical guarantees demonstrating that, under specific structural conditions on the feature set, the proposed algorithm achieves asymptotic exact surface recovery independently of the applied noise level.

\textbf{Notation:} Scalars, vectors, and matrices are denoted by non-bold ($x$), lowercase bold ($\mathbf{x}$), and uppercase bold ($\mathbf{X}$) letters, respectively, with $x_i$ and $X_{i,j}$ representing their individual entries. Let $\mathbb{R}$ and $\mathbb{Z}_{+}$ denote the sets of real numbers and positive integers. For a set $\mathcal{X}$, $|\mathcal{X}|$ signifies its cardinality, and for $n \in \mathbb{Z}_{+}$, $[n] \stackrel{\Delta}{=} \{1,\hdots,n\}$. Matrices $\mathbf{1}_{n\times m}$ and $\mathbf{0}_{n\times m}$ are all-ones and all-zeros matrices, and $\mathbf{e}_{i}$ is the $i$-th standard basis vector. The expected value operator is denoted by $\EX[\cdot]$. Euclidean and Frobenius norms are both represented by $\|\cdot\|$. The operator $diag(\mathbf{x})$ forms a diagonal matrix from the elements of vector $\mathbf{x}$, whereas $diag\{\mathbf{X}_{1}, \dots, \mathbf{X}_{k}\}$ constructs a block diagonal matrix from the provided matrices. Vertical matrix concatenation is denoted as $vercat\{\mathbf{X}_{1},\mathbf{X}_{2}\}$. Finally, $\sigma_{\mathrm{min}}\{\mathbf{X}\}$, $\mathbf{v}(\sigma_{\mathrm{min}}\{\mathbf{X}\})$, and $\mathcal{N}(\mathbf{X})$ represent the minimum singular value, its associated singular vector, and the null space of matrix $\mathbf{X}$, respectively.

 \section{Surface Learning}\label{sec:Surf}

Surface learning focuses on modeling the underlying structure of data by finding a mathematical representation of a surface that best fits the data points. As mentioned in the introduction, such a mathematical representation is essential if one is to extrapolate the information provided by available data. The proposed approach focuses on estimating the underlying surface from point clouds that are significantly impacted by high noise levels. More precisely, we aim at problems where the data are representable by the level set of a function belonging to the span of a given set of features/basis.

Consider a set of features $\mathbf{b}(\mathbf{x}) \stackrel{\Delta}{=}[b_{1}(\mathbf{x}), \dots, b_{N}(\mathbf{x})]^{\top}$, where $b_{i}: \mathbb{R}^{n} \rightarrow \mathbb{R}$, $\forall i \in [N]$. Moreover, let ${\mathcal{D}}_{L} \subset \mathbb{R}^{n}$ be a set of $L$ points on the surface and $\Tilde{\mathcal{D}}_{L} \subset \mathbb{R}^{n}$ be the set of $L$ noisy measurements; i.e., $\mathbf{y} \in \Tilde{\mathcal{D}}_{L}$ if and only if there exists an $\mathbf{x} \in {\mathcal{D}}_{L}$ and an allowable noise value $\boldsymbol{\epsilon} \in \mathbb{R}^{n}$ such that $\mathbf{y} = \mathbf{x} + \boldsymbol{\epsilon}$. Then, the problem that we aim at addressing can be stated as:

\begin{problem}
Given a set of noisy measurements $\Tilde{\mathcal{D}}_{L}$, find $\mathbf{a} \in \mathbb{R}^{N}$ such that $\mathbf{a}^{\top} \mathbf{b}(\mathbf{x}) = 0$ for all $\mathbf{x} \in {\mathcal{D}}_{L}$.
\end{problem}

        \subsection{Surface Learning with Noiseless Data} \label{sec:noiseless}

To address this problem, we first consider the noiseless scenario, where the surface fitting objective is formulated as:
            \begin{equation} 
            \begin{aligned}
                \text{Find} \quad \mathbf{a}\in \mathbb{R}^{N}, \quad
                \textrm{s.t.} \quad \mathbf{a}^{\top}\mathbf{M}_{\mathcal{D}_{L}}\mathbf{a}=0,
                \end{aligned}
            \end{equation}\label{eq:coeff_vector}
            where $\mathbf{M}_{\mathcal{D}_{L}}=\frac{1}{L}\sum_{\mathbf{x}\in\mathcal{D}_{L}}\mathbf{M}(\mathbf{x}) \text{ and } \mathbf{M}(\cdot)=\mathbf{b}(\cdot)\mathbf{b}(\cdot)^{\top}.$

            Hence, the noiseless implicit surface fitting problem is equivalent to finding the coefficient vector $\mathbf{a}$ that belongs to the null space of the matrix $\mathbf{M}_{\mathcal{D}_{L}}$. The matrix $\mathbf{M}_{\mathcal{D}_{L}}$ is just the sum of terms involving the evaluation of features at the measured points, a calculation that is linear in the number of points $L$. Then, one just needs to find an eigenvector of a matrix whose size only depends on the number of features.

\subsection{Challenges Imposed by Noisy Data in Surface Learning}
     Now, let us go back to the case where  only noisy realizations $\mathbf{y}=\mathbf{x}+\boldsymbol\epsilon\in \Tilde{\mathcal{D}}_{L} \subset \mathbb{R}^n$ are available. First, note that under suitable assumptions on the noise and for a sufficiently large number of measurements $L$, we have 
    \[
\frac{1}{L}\sum_{\mathbf{y}\in\Tilde{\mathcal{D}}_{L}}\mathbf{M}(\mathbf{y}) \approx \frac{1}{L}\sum_{\mathbf{y}\in\Tilde{\mathcal{D}}_{L}}\EX[\mathbf{M}(\mathbf{y})]
    \]
 and we can think of \(\frac{1}{L}\sum_{\mathbf{y}\in\Tilde{\mathcal{D}}_{L}}\EX[\mathbf{M}(\mathbf{y})]\) as a surrogate for $\mathbf{M}_{\mathcal{D}_{L}}$ and aim at estimating the surface as the solution of  
%
    \begin{equation}
        \begin{aligned}
        \text{Find} \quad \mathbf{a}\in \mathbb{R}^{N}, \quad
        \textrm{s.t.} \quad \mathbf{a}^{\top}\Big(\frac{1}{L}\sum_{\mathbf{y}\in\Tilde{\mathcal{D}}_{L}}\EX[\mathbf{M}(\mathbf{y})]\Big)\mathbf{a}=0.
        \end{aligned}
    \end{equation}
    However, since $\mathbf{M}(\mathbf{y})$ depends on the structure of the feature functions within the vector $\mathbf{b(\cdot)}$ and the distribution of the added noise, there is no guarantee that $\frac{1}{L}\sum_{\mathbf{y}\in\Tilde{\mathcal{D}}_{L}}\EX[\mathbf{M}(\mathbf{y})]=\mathbf{M}_{\mathcal{D}_{L}}$, i.e., it can be a biased estimate of $\mathbf{M}_{\mathcal{D}_{L}}$.

The subsequent section details how RAIN-FIT compensates for this bias to accurately estimate the target surface.

    \section{RAIN-FIT Algorithm}\label{sec:RAIN}
    
In essence, RAIN-FIT tackles the challenge of fitting surfaces to high-level noisy data, a common issue in real-world data. RAIN-FIT stands out by offering robust noise resistance while keeping the computation lightweight. 

\subsection{Motivation}

Consider the following example which provides a first example of the structure exploited by the approach proposed in this paper.

    \begin{example} \label{ex:ex2}
        Let $\mathcal{D}_{L}\subseteq \mathbb{R}^{2}$ and $\mathbf{b}(\mathbf{x})=[x_{1}, x_{2}]^{\top}$. The expected value of the matrix $\mathbf{M}(\mathbf{y})$ can be shown to be:
        \begin{equation} \label{eq:biased}
            \begin{aligned}
                &\EX[\mathbf{M}(\mathbf{y})]=\mathbf{M}(\mathbf{x})+\EX[\mathbf{B}(\mathbf{y}, f_{\boldsymbol\theta})],
            \end{aligned}
        \end{equation}
        where we used the fact that $x_{k}= \EX[y_{k}]-\EX[\epsilon_{k}]$ and $\mathbf{B}(\mathbf{y},f_{\boldsymbol\theta})$ is defined as follows:

\begin{equation} \label{eq:bias_def}
    \mathbf{B}(\mathbf{y}, f_{\boldsymbol\theta}) = \begin{bmatrix}
        B_{11} & B_{12} \\
        B_{21} & B_{22}
    \end{bmatrix},
\end{equation}
where the elements are given by:
\begin{align*}
    B_{11} &= 2y_{1}\EX[\epsilon_{1}] - 2(\EX[\epsilon_{1}])^{2} + \EX[\epsilon_{1}^{2}], \\
    B_{12} &= y_{1}\EX[\epsilon_{2}] + y_{2}\EX[\epsilon_{1}] - \EX[\epsilon_{1}]\EX[\epsilon_{2}], \\
    B_{21} &= y_{1}\EX[\epsilon_{2}] + y_{2}\EX[\epsilon_{1}] - \EX[\epsilon_{1}]\EX[\epsilon_{2}], \\
    B_{22} &= 2y_{2}\EX[\epsilon_{2}] - 2(\EX[\epsilon_{2}])^{2} + \EX[\epsilon_{2}^{2}].
\end{align*}
        It is evident from \eqref{eq:bias_def} that the bias matrix $\mathbf{B}(\mathbf{y}, f_{\boldsymbol\theta})$ is equal to zero when the first and second moments of the noise are both zero. In a more general context, the bias matrix's values are contingent on the specific characteristics of the random vectors $\mathbf{y}$ and $\boldsymbol\epsilon$, as well as the underlying feature functions, within the vector $\mathbf{b}(\cdot)$.
    \end{example}
    
As illustrated, the structure of the feature functions allows $\EX[\mathbf{M}(\mathbf{y})]$ to be additively decomposed into the true noiseless matrix $\mathbf{M}(\mathbf{x})$ and a stochastic bias term $\EX[\mathbf{B}(\mathbf{y}, f_{\boldsymbol\theta})]$. Because direct computation over noisy observations inherently yields biased estimates, isolating this noise-induced bias is critical. This explicit separation provides the mathematical foundation to actively compensate for noise and accurately recover the underlying surface.

     \subsection{Conditions on Feature Functions}

         To have a consistent way of defining the elements and dimension of the feature vector $\mathbf{b(\cdot)}$, we define the model order $\gamma \in \mathbb{Z}_{+}$, such that $\gamma$ indicates a family of ordered functions, $\mathcal{C}_{\gamma}=\{b_{1}^{(\gamma)}, b_{2}^{(\gamma)}, \dots, b_{N_{\gamma}}^{(\gamma)}\}$, that share a common parameter $\gamma$, where $b_{i}^{(\gamma)}:\mathbb{R}^{n}\rightarrow \mathbb{R}$, $\forall$ $i\in [N_{\gamma}]$ lies on the set of interest. For a given model order $\gamma$ and the corresponding family definition $\mathcal{C}_{\gamma}$, the feature vector:
        \begin{equation} \label{eq:basis_vector_def}
            \begin{aligned}
                \mathbf{b}(\mathbf{x})&=[\underbrace{b_{1}^{(0)}(\mathbf{x}),\dots b_{N_{0}}^{(0)}(\mathbf{x})}_{\mathcal{C}_{0}},\dots,  \underbrace{b_{1}^{(\gamma)}(\mathbf{x}),\dots, b_{N_{\gamma}}^{(\gamma)}(\mathbf{x})}_{\mathcal{C}_{\gamma}}]^{\top} \\
                &= [\mathbf{b}^{(0)}(\mathbf{x})^{\top},\dots, \mathbf{b}^{(j)}(\mathbf{x})^{\top}, \dots, \mathbf{b}^{(\gamma)}(\mathbf{x})^{\top}]^{\top},
            \end{aligned}
        \end{equation}
        is defined as the vector of all the functions in $\cup_{k=0}^{\gamma}\mathcal{C}_{k}$ with $N=\sum_{k=0}^{\gamma}N_{k}$, and $\mathbf{b}^{(j)}(\cdot)$ is the vector of $N_{\gamma}$ feature functions in $\mathcal{C}_{\gamma}$. Hence, the pair $(\gamma, \mathcal{C}_{\gamma})$ provides the complete information to construct the feature vector $\mathbf{b}(\mathbf{x})$. We make the following assumption: 
        \begin{assumption}
            For the feature function $b_{i}^{(j)}(\mathbf{x})\in \mathbf{b}(\mathbf{x})$, where $i\in [N_{j}]$, and $j\in [\gamma]$, the following property is satisfied:
            \begin{equation} \label{eq:basis_prop}
                b_{i}^{(j)}(\mathbf{x}+\boldsymbol\epsilon)=\sum_{m=0}^{j}\sum_{k=1}^{N_{m}}\sum_{\ell=0}^{\gamma} \sum_{r=1}^{N_{\ell}}  c_{k,r,i}^{(m,\ell,j)}b_{k}^{(m)}(\mathbf{x})b_{r}^{(\ell)}(\boldsymbol\epsilon).
            \end{equation}    
            Moreover, for any functions denoted as $b_{*}^{(j)}(\mathbf{x})\in\mathcal{C}_{j}$ and $b_{*}^{(\ell)}(\mathbf{x})\in\mathcal{C}_{\ell}$, it follows that:
            \begin{equation} \label{eq:basis_mult}
                b_{*}^{(j)}(\mathbf{x})b_{*}^{(\ell)}(\mathbf{x})\in\text{Span}\left\{\bigcup_{k=0}^{j+\ell}\mathcal{C}_{k}\right\}.   
            \end{equation}
            In other words, the product of $b_{*}^{(j)}(\mathbf{x})$ and $b_{*}^{(\ell)}(\mathbf{x})$ is an element of the span of functions encompassing $\mathcal{C}_{j+\ell}$ and functions of lower orders. \label{A6}  
        \end{assumption} 

  Property~\eqref{eq:basis_prop} establishes a separation property: any feature evaluated at a noisy point decomposes into a linear combination of noiseless data features weighted by noise features. This is an intrinsic mathematical characteristic of standard (not necessarily orthogonal) bases, such as polynomial, trigonometric, and exponential functions. Furthermore, \eqref{eq:basis_mult} ensures that the product of any two features remains within a span that also satisfies \eqref{eq:basis_prop}.

\begin{remark}
    Condition~\eqref{eq:basis_prop} isolates $b_{i}^{(j)}(\mathbf{x})$ as a separable function of the noisy observations and the noise distribution. Coupled with \eqref{eq:basis_mult}, this guarantees that the entries of the matrix $\mathbf{M}(\mathbf{y}) = \mathbf{b}(\mathbf{y})\mathbf{b}(\mathbf{y})^{\top}$ inherit this separation.
\end{remark}

Rather than assuming complete knowledge of the noise distribution, we assume it is known only up to a finite set of parameters. Formally, we posit the following:
    
             \begin{assumption} \label{assump:noise}
            The noise vector elements $\epsilon_{i}$ and $\epsilon_{l}$ are independent for $i \neq l$ and identically distributed with a parameterized probability density $f_{\boldsymbol\theta}(\boldsymbol\epsilon)$, where the parameter $\boldsymbol\theta$ is a low-dimensional vector. Therefore, the form of the noise probability distribution $f_{\boldsymbol\theta}$ is known, while the parameter vector $\boldsymbol\theta$ is not. \label{A1}
        \end{assumption}

    \begin{remark}
    Given the RAIN-FIT's low computational complexity, Assumption~\ref{assump:noise} is not practically restrictive. One can efficiently evaluate a finite set of candidate noise distributions and their parameters.
    \end{remark}

     Given the above assumptions, we can now precisely define the problem as follows:
    
        \begin{problem}
            Given $\Tilde{\mathcal{D}}_{L}$, under Assumptions \ref{A6}-\ref{assump:noise}, simultaneously estimate the noiseless surface and the unknown noise parameters $\boldsymbol\theta$.
        \end{problem}

    \subsection{Noise Compensation with Ordered Feature Sets}
        
        In this subsection, we present a procedure to address and rectify potential bias in the estimation of $\mathbf{M}_{\mathcal{D}_{L}}$. This corrective process, herein referred to as \textit{noise compensation}, capitalizes on the structured (ordered) nature of the feature functions within $\mathbf{b}(\mathbf{x})$, and their compliance with \eqref{eq:basis_prop}.

        Exploiting this ordered characteristic of the feature functions, we can establish a relation between $b_{i}^{(j)}(\mathbf{x})$ and lower-ordered feature functions substituted at the available noisy realizations. As we shall demonstrate, this approach enables the computation of the bias introduced into $\mathbf{M}(\mathbf{x})$ and consequently, eliminates it. 
        
        As a first step, we split all the terms with functions in $\mathcal{C}_{j}$, i.e., $b_{*}^{(j)}(\mathbf{x})$, in \eqref{eq:basis_prop} and perform the expected value with respect to $\epsilon$, which yields to:
        \begin{equation} \label{eq:exp_with_sep}
            \begin{aligned}
                &\EX[b_{i}^{(j)}(\mathbf{x}+\boldsymbol\epsilon)]=\sum_{k=1}^{N_{j}}b_{k}^{(j)}(\mathbf{x})\sum_{\ell=0}^{\gamma}\sum_{r=1}^{N_{\ell}}  c_{k,r,i}^{(j,\ell,j)}\EX[b_{r}^{(\ell)}(\boldsymbol\epsilon)] \\
                &+\sum_{m=0}^{j-1}\sum_{k=1}^{N_{m}}b_{k}^{(m)}(\mathbf{x})\sum_{\ell=0}^{\gamma}\sum_{r=1}^{N_{\ell}}  c_{k,r,i}^{(m,\ell,j)}\EX[b_{r}^{(\ell)}(\boldsymbol\epsilon)], \quad \forall i\in N_{j}.
            \end{aligned}
        \end{equation}
        Let the matrix $\mathbf{A}(m, j; f_{\boldsymbol\theta})\in\mathbb{R}^{N_{m}\times N_{j}}$ be defined such that its $ki$-entry is:
        \begin{equation} \label{eq:A_matrix}
            A_{k,i}(m, j; f_{\boldsymbol\theta})\stackrel{\Delta}{=}\sum_{\ell=0}^{\gamma}\sum_{r=1}^{N_{\ell}}c_{k,r,i}^{(m,\ell,j)}\EX[b_{r}^{(\ell)}(\boldsymbol\epsilon)],    
        \end{equation}
        therefore, \eqref{eq:exp_with_sep} can be represented in the matrix form as:
        \begin{equation} \label{eq:mat_form}
            \begin{aligned}
                \EX[\mathbf{b}^{(j)}(\mathbf{x}+\boldsymbol\epsilon)] &= \mathbf{A}^{\top}(j, j; f_{\boldsymbol\theta})\mathbf{b}^{(j)}(\mathbf{x}) \\
                &+\sum_{m=0}^{j-1}\mathbf{A}^{\top}(m, j; f_{\boldsymbol\theta})\mathbf{b}^{(m)}(\mathbf{x}),
            \end{aligned}
        \end{equation}
        hence, from \eqref{eq:mat_form}, $\mathbf{b}^{(j)}(\mathbf{x})$ can be calculated as follows:
        \begin{equation}\label{eq:bi_calc}
            \begin{aligned}
                \mathbf{b}^{(j)}(\mathbf{x})=\Big[&\mathbf{A}^{\top}(j, j; f_{\boldsymbol\theta})\Big]^{-1}\Big[\EX[\mathbf{b}^{(j)}(\mathbf{x}+\boldsymbol\epsilon)] \\
                &-\sum_{m=0}^{j-1}\mathbf{A}^{\top}(m, j; f_{\boldsymbol\theta})\mathbf{b}^{(m)}(\mathbf{x})\Big].
            \end{aligned}
        \end{equation}

            It becomes evident that the calculation of $\mathbf{b}^{(j)}(\mathbf{x})$ as described in \eqref{eq:bi_calc} is dependent on $\mathbf{b}^{(m)}(\mathbf{x})$ for $\forall m < j$. Using \eqref{eq:bi_calc}, $\mathbf{b}^{(m)}(\mathbf{x})$ can be iteratively computed while the initial condition is considered to be $\mathbf{b}^{(0)}(\mathbf{x})=1$.
        The subtraction operation in \eqref{eq:bi_calc}, serves to perform the noise compensation as it implicitly rectifies any bias introduced during the estimation of $\mathbf{M}(\mathbf{x})$ when based on noisy realizations.

                \begin{remark}
            The expression in \eqref{eq:bi_calc} depends on the knowledge of the introduced noise distribution following the premise outlined in Assumption \ref{A1}.
        \end{remark}
        
        The selection of feature functions encompassed within the vector $\mathbf{b}(\cdot)$ plays a pivotal role in guaranteeing the invertibility of the matrix $\mathbf{A}^{\top}(j, j; f_{\boldsymbol\theta})$.
        \begin{assumption}
            The feature functions are chosen in a manner that ensures the invertibility of $\mathbf{A}^{\top}(j, j; f_{\boldsymbol\theta})$. \label{A7}
        \end{assumption}

        A diverse array of functions conforming to the conditions stipulated in \eqref{eq:basis_prop} and \eqref{eq:basis_mult} exist, thus affording the possibility to identify a feature vector that ensures the invertibility of the matrix $\mathbf{A}^{\top}(j, j; f_{\boldsymbol\theta})$.

        \begin{example} \label{ex:example3}
            Let $\mathbf{b}(x):\mathbb{R}\rightarrow \mathbb{R}^{N}$ be defined by the pair $(\gamma, \mathcal{C}_{\gamma})$, where $\mathcal{C}_{\gamma}=\{x^{\gamma}\}$. Correspondingly, $\mathbf{b}^{(\gamma)}(x):\mathbb{R}\rightarrow \mathbb{R}$, where $\mathbf{b}^{(j)}(x)=x^{j}$. It is then easy to realize that $\EX[\mathbf{b}^{(j)}(x+\epsilon)]$ can be expanded using the binomial theorem as:
            \begin{equation} \label{ex:monomial}
                \begin{aligned}
                    \EX[\mathbf{b}^{(j)}(x+\epsilon)] = x^{j} + \sum_{u=1}^{j} \binom{j}{u} x^{j-u}\EX[\epsilon^{u}],
                \end{aligned}
            \end{equation}
            and hence, $\mathbf{b}^{(j)}(x)$ can be iteratively computed from its noisy realization $\mathbf{b}^{(j)}(x+\epsilon)$ and its lower order estimates as follows:
            \begin{equation} \label{eq:xj_estimate}
                x^{j} = \EX[\mathbf{b}^{(j)}(x+\epsilon)] - \sum_{u=1}^{j} \binom{j}{u} x^{j-u}\EX[\epsilon^{u}]
            \end{equation}
            In order to see \eqref{ex:monomial} in the light of \eqref{eq:basis_prop}, since in that case $N_{m}=N_{\ell}=1$ and $|\mathcal{C}_{N}|=1$, the indices $i$, $k$ and $r$ can be dropped. This yields: 
            \begin{equation}
                \mathbf{b}^{(j)}(x+\epsilon)=\sum_{m=0}^{j}\sum_{\ell=0}^{\gamma}c^{(m,\ell,j)}\mathbf{b}^{(m)}(x)\mathbf{b}^{(\ell)}(\epsilon).
            \end{equation}
            By letting $c^{(m,\ell,j)}=0$ for every $m+\ell\neq j$ we get that:
            \begin{equation} \label{ex:monomial_mapped}
                \begin{aligned}
                    &\EX[\mathbf{b}^{(j)}(x+\epsilon)]=\sum_{u=0}^{j}w_{u}^{(j)}\mathbf{b}^{(j-u)}(x)\EX[\mathbf{b}^{(u)}(\epsilon)] \\
                    &=w_{0}^{(j)}\mathbf{b}^{(j)}(x)\EX[\mathbf{b}^{(0)}(\epsilon)]+\sum_{u=1}^{j}w_{u}^{(j)}\mathbf{b}^{(j-u)}(x)\EX[\mathbf{b}^{(u)}(\epsilon)].
                \end{aligned}
            \end{equation}
            It can be then realized that \eqref{ex:monomial_mapped} can be mapped to \eqref{ex:monomial} by letting the constant $w_{u}^{(j)}=\binom{j}{u}$ and with the fact that $\EX[\mathbf{b}^{(0)}(\epsilon)]=1$.
        \end{example}

While Example \ref{ex:example3} demonstrates the estimation of $\mathbf{b}^{(j)}(\mathbf{x})$ for monomial features (where $|\mathcal{C}_{\gamma}|=1$ for all $\gamma$), our broader objective is the computation of the full matrix $\mathbf{M}(\mathbf{x})=\mathbf{b}(\mathbf{x})\mathbf{b}(\mathbf{x})^{\top}$. By Assumption \ref{A6}, the elements of the sub-matrix $b^{(j)}(\mathbf{x})b^{(\ell)}(\mathbf{x})^{\top}$ belong to $\mathcal{C}_{j+\ell}$ and satisfy \eqref{eq:basis_prop}, enabling their computation via the analogous iterative process in \eqref{eq:bi_calc}.

Thus far, evaluating \eqref{eq:bi_calc} requires the exact expectation $\EX[\mathbf{b}^{(j)}(\mathbf{y})]$. In practice, this analytical value is unattainable because the true coordinates $\mathbf{x}$ are unknown and noise distribution knowledge is only partial. To overcome this limitation, we replace the theoretical expectation with an empirical sample mean estimate, yielding the following modification to \eqref{eq:bi_calc}:
    \begin{equation} \label{eq:sample_mean_est}
        \begin{aligned}
            \hat{\mathbf{b}}^{(j)}(\mathbf{x})=\Big[&\mathbf{A}^{\top}(j, j; f_{\boldsymbol\theta})\Big]^{-1}\Big[\mathbf{b}^{(j)}(\mathbf{y}) \\
            &-\sum_{m=0}^{j-1}\mathbf{A}^{\top}(m, j; f_{\boldsymbol\theta})\hat{\mathbf{b}}^{(m)}(\mathbf{x})\Big].
        \end{aligned}
    \end{equation}
    Applying \eqref{eq:sample_mean_est}, the estimate $\Hat{\mathbf{M}}(\mathbf{y})$ yields:
    \begin{equation} 
    \label{eq:bias_cancel_mat}
        \begin{aligned}
            \Hat{\mathbf{M}}(\mathbf{y}, f_{\theta}) = \mathbf{M}(\mathbf{y}) - \mathbf{B}(\mathbf{y}, f_{\boldsymbol\theta}), 
        \end{aligned}
    \end{equation}
    and hence, we obtain the matrix estimate $\Hat{\mathbf{M}}_{\mathcal{D}_{L}}=\frac{1}{L}\sum_{\mathbf{y}\in\Tilde{\mathcal{D}}_{L}}\Hat{\mathbf{M}}(\mathbf{y}, f_{\boldsymbol\theta})$.
    Analogous to the formulation in \eqref{eq:coeff_vector}, we derive the coefficient vector $\hat{\mathbf{a}}_{L}$ by identifying the singular vector $\mathbf{v}$ corresponding to the minimum singular value $\lambda_{\mathrm{min}}$ of the matrix $\Hat{\mathbf{M}}_{\mathcal{D}_{L}}$. 

\begin{remark}
The RAIN-FIT framework naturally extends to arbitrary analytic feature sets. By approximating these features using standard bases (e.g., polynomial, trigonometric, or exponential functions) over compact domains, the proposed algorithm can seamlessly adapt to virtually any analytic feature space.
\end{remark}

\begin{algorithm}[t]
    \caption{One-Time Offline Calculation of $\hat{\mathbf{M}}(., f_{.})$ }
    \label{alg:gen_toolbox}
    \renewcommand{\thealgorithm}{}
    \floatname{algorithm}{}
        \begin{algorithmic}[1]
            \State Given: The pair $(\gamma, \mathcal{C}_{\gamma})$ and the noise distribution $f$. 
            \State Construct the $N$-dimensional vector of functions $\mathbf{b}(\cdot) = [\underbrace{1}_{\mathcal{C}_{0}}, \underbrace{\mathbf{b}^{(1)}(\cdot)}_{\mathcal{C}_{1}}, \underbrace{\mathbf{b}^{(2)}(\cdot)}_{\mathcal{C}_{2}}, \dots \underbrace{\mathbf{b}^{(\gamma)}(\cdot)}_{\mathcal{C}_{\gamma}}]^{\top}$.
            \State Construct $\mathbf{M}(\cdot)=\mathbf{b}(\cdot)\mathbf{b}(\cdot)^{\top}$ with the $k\ell$-th entry as the function $M_{k\ell}(\cdot)$.
            \State Let the matrix $\hat{\mathbf{M}}(., f_{.})$ be an estimate of $\mathbf{M}(\cdot)$.
            \State Determine $\theta$ via grid search over $N$ points, selecting $\theta$ that stabilizes singular values of $\hat{\mathbf{M}}$.
            \For{$k, \ell\in[N]$}
                    \State Using Assumption \ref{A6} and \eqref{eq:sample_mean_est}, derive the functional form of $\hat{M}_{k\ell}(., f_{.})$, the $k\ell$-th entry of $\hat{\mathbf{M}}(., f_{.})$.
            \EndFor
        \end{algorithmic}
\end{algorithm}

    
To avoid computational bottlenecks, we recognize that $\hat{\mathbf{M}}(\mathbf{y}, f_{\boldsymbol\theta})$ is independent of the specific dataset $\Tilde{\mathcal{D}}_{L}$ and can be derived as a generalized function $\hat{\mathbf{M}}(\cdot, f_{\cdot})$ dependent only on the noisy sample $\mathbf{y}$ and noise parameter $\boldsymbol\theta$. Consequently, our approach decouples into two phases. First, a one-time phase (Algorithm \ref{alg:gen_toolbox}) constructs the feature vector $\mathbf{b}(\cdot)$ and analytically derives the functional form of each matrix entry $\hat{M}_{k\ell}(\cdot, f_{\cdot})$ via Assumption \ref{A6} and \eqref{eq:sample_mean_est}. Second, the RAIN-FIT phase evaluates this pre-computed functional form across the specific dataset $\Tilde{\mathcal{D}}_{L}$. The global estimate $\Hat{\mathbf{M}}_{\mathcal{D}_{L}}$ is efficiently computed by averaging these sample-wise matrices, yielding the optimal surface coefficient vector via the singular vector corresponding to the minimum singular value of $\Hat{\mathbf{M}}_{\mathcal{D}_{L}}$.
\begin{remark}
    The computational complexity of RAIN-FIT is governed by the sample size $L$ and the number of features $N$. Because the functional form of $\hat{M}_{k\ell}(\cdot, f_{\cdot})$ is derived (Algorithm \ref{alg:gen_toolbox}), the complexity is strictly determined by matrix averaging and singular value decomposition (SVD). Assuming $t$ is the time required to evaluate each $\hat{M}_{k\ell}(\mathbf{y}, f_{\boldsymbol\theta})$, computing the $N \times N$ matrix over $L$ samples takes $\mathcal{O}(L(t+N^2))$. Incorporating the $\mathcal{O}(N^3)$ SVD step, the total computational complexity is strictly $\mathcal{O}(L(t+N^{2})+N^{3})$.
\end{remark}
    \begin{algorithm}[t]
        \caption{RAIN-FIT} 
        \label{alg:cal_a}
        \renewcommand{\thealgorithm}{}
        \floatname{algorithm}{}
            \begin{algorithmic}[1]
                \State Given: Data set $\Tilde{\mathcal{D}}_{L}$ of noisy samples and the noise parameter vector $\boldsymbol\theta$.
                \State Calculate $\Hat{\mathbf{M}}_{\mathcal{D}_{L}}=\frac{1}{L}\sum_{\mathbf{y}\in\Tilde{\mathcal{D}}_{L}}\hat{\mathbf{M}}(\mathbf{y}, f_{\boldsymbol\theta})$.
                \State The coefficients vector for implicit surface fitting is $\hat{\mathbf{a}}_{L} = \mathbf{v}(\sigma_{\mathrm{min}}\{\Hat{\mathbf{M}}_{\mathcal{D}_{L}}\})$.
            \end{algorithmic}
    \end{algorithm}


    \section{Algorithm Convergence Analysis}\label{sec:Conv}
    To begin analyzing the convergence of RAIN-FIT, we first state a basic assumption that is used throughout this section.

    \begin{assumption}
    Each data point $\mathbf{x} \in \mathcal{D}_{L}$ can be represented by the zero set of an implicit mathematical function $g(\mathbf{x}) = \textbf{a}^T\mathbf{b}(\mathbf{x})$, where $\textbf{a} \in \mathbb{R}^N$. The function $g$, whose zero set includes the original point cloud, can be expressed as a linear combination of the elements in the feature vector $\mathbf{b}(\mathbf{x})$. Moreover, throughout the paper, it is assumed that for any feature function $b_{i}(\cdot) \in \mathbf{b}(\cdot)$, we let:

    \begin{itemize}
        \item $|b_{i}(\mathbf{x})| \leq c_{1} < \infty$ for any $\mathbf{x} \in \mathcal{D}_{L}$ and $\forall i \in [N]$.
        \item $\EX[b_{i}(\boldsymbol\epsilon)]  \leq c_{2} < \infty$.
        \item $\EX[(b_{i}(\boldsymbol\epsilon)-\EX[b_{i}(\boldsymbol\epsilon)])^{2}] \leq c_{3} <\infty$.
    \end{itemize}\label{A2}
\end{assumption}

    \begin{lemma}
        $\Hat{\mathbf{M}}_{\mathcal{D}_{L}}$ generated through RAIN-FIT exhibits the following entry-wise convergence property:
        \begin{equation} \label{eq:limit_relation}
            \lim_{L\to\infty}\Hat{\mathbf{M}}_{\mathcal{D}_{L}}-\mathbf{M}_{\mathcal{D}_{L}}\rightarrow 0, \quad \text{a.s.}
        \end{equation}
    \end{lemma}
    \begin{proof}
        Please find the proof in \cite{sleem2026rain}.
    \end{proof}

    Given Assumptions \ref{A1} and \ref{A2}, it can be inferred that the null space of matrix $\mathbf{M}_{\mathcal{D}_{L}}$ is consistent with dimension one. Nevertheless, there exist two optimal vectors, denoted as $\mathbf{a}^{*}$ and $-\mathbf{a}^{*}$. As a result, we define a distance measure:
    \begin{equation} \label{eq:distance}
        d_{L}=\min_{\boldsymbol\xi\in\{\mathbf{a}^{*},-\mathbf{a}^{*}\}} \|\hat{\mathbf{a}}_{L}-\boldsymbol\xi\|.
    \end{equation}
    
    Finally, we can establish the following theorem:
    \begin{theorem}\label{thm:convergence}
       Under Assumptions \ref{A6}-\ref{A2}, as $L$ tends to infinity, the sequence $d_{L}$ generated by \eqref{eq:distance} converges almost surely to zero. In other words, the distance between the vector $\hat{\mathbf{a}}_{L}$, generated through RAIN-FIT, and the set containing the null space vectors, $\mathbf{a}^{*}$ and $-\mathbf{a}^{*}$, of the matrix $\mathbf{M}_{\mathcal{D}_{L}}$ converges almost surely to zero.
    \end{theorem}
    \begin{proof}
        Please find the proof in \cite{sleem2026rain}.     
    \end{proof}
    Exponential convergence of the entries of the matrix $\hat{\mathbf{M}}(\mathbf{y}, f_{\boldsymbol\theta})$ follows from results on the Central Limit Theorem. Specifically, according to Hoeffding's inequality \cite[Equation 2.3.17]{sen1994large}, when the random variables within the entries of $\hat{\mathbf{M}}(\mathbf{y}, f_{\boldsymbol\theta})$ are bounded for each $\mathbf{y} \in \Tilde{\mathcal{D}}_{L}$, RAIN-FIT achieves an exponential convergence rate.

\section{Smooth representation of data} \label{sec:smoothing}
    The preceding analysis relies on Assumption \ref{A1}, strictly confining data points to a zero set. Violating this yields irregular surfaces, as the solution to \eqref{eq:coeff_vector} fails to minimize orthogonal geometric distances \cite{841760}. To enforce smoothness, we adapt the "3L" method \cite{841760}, restricting $\mathbf{a}^{*}$ to spaces free of local extrema near the data by bounding it within a "ribbon-like" surface \cite{DANIELSSON1980227}. Denoting the augmented sets as $\mathcal{D}_{L}$, $\mathcal{D}_{L}^{-c}$, and $\mathcal{D}_{L}^{+c}$, let $\boldsymbol\Gamma_{\mathcal{D}_{L}}$, $\boldsymbol\Gamma_{\mathcal{D}_{L}^{-c}}$, and $\boldsymbol\Gamma_{\mathcal{D}_{L}^{+c}} \in \mathbb{R}^{L\times N}$ represent their respective matrices of feature evaluations $\mathbf{b}(\mathbf{x})^{\top}$. The 3L method solves:
    \begin{equation} \label{eq:3L_matrix_system}
        \begin{bmatrix}
            \boldsymbol\Gamma_{\mathcal{D}_{L}^{-c}} \\
            \boldsymbol\Gamma_{\mathcal{D}_{L}} \\
            \boldsymbol\Gamma_{\mathcal{D}_{L}^{+c}}
        \end{bmatrix} \mathbf{a} = 
        c\begin{bmatrix}
            \mathbf{-1}_{L\times 1} \\
            \mathbf{0}_{L\times 1} \\
            \mathbf{1}_{L\times 1}
        \end{bmatrix},
    \end{equation}
    where the solution is simply the pseudo inverse. 

    To integrate this approach with RAIN-FIT's matrix averaging framework, we reformulate the system in \eqref{eq:3L_matrix_system} into the equivalent form:
    \begin{equation}
        \begin{aligned} 
            &\begin{bmatrix}
                1 & \mathbf{a}^{\top}
            \end{bmatrix}
            \underbrace{\begin{bmatrix}
                c \\
                \mathbf{b}(\mathbf{x})
            \end{bmatrix}
            \begin{bmatrix}
                c & \mathbf{b}(\mathbf{x})^{\top}
            \end{bmatrix}}_{\stackrel{\Delta}{=}\Tilde{\mathbf{M}}_{-c}(\mathbf{x})}
            \begin{bmatrix}
                1 \\
                \mathbf{a}
            \end{bmatrix} = 0, \quad \forall \mathbf{x}\in \mathcal{D}_{L}^{-c},   \\
            &\begin{bmatrix}
                1 & \mathbf{a}^{\top}
            \end{bmatrix}
            \underbrace{\begin{bmatrix}
                0 \\
                \mathbf{b}(\mathbf{x})
            \end{bmatrix}
            \begin{bmatrix}
                0 & \mathbf{b}(\mathbf{x})^{\top}
            \end{bmatrix}}_{\stackrel{\Delta}{=}\Tilde{\mathbf{M}}_{0}(\mathbf{x})}
            \begin{bmatrix}
                1 \\
                \mathbf{a}
            \end{bmatrix} = 0, \quad \forall \mathbf{x}\in \mathcal{D}_{L},\\
            &\begin{bmatrix}
                1 & \mathbf{a}^{\top}
            \end{bmatrix}
            \underbrace{\begin{bmatrix}
                -c \\
                \mathbf{b}(\mathbf{x})
            \end{bmatrix}
            \begin{bmatrix}
                -c & \mathbf{b}(\mathbf{x})^{\top}
            \end{bmatrix}}_{\stackrel{\Delta}{=}\Tilde{\mathbf{M}}_{+c}(\mathbf{x})}
            \begin{bmatrix}
                1 \\
                \mathbf{a}
            \end{bmatrix} = 0, \quad \forall \mathbf{x}\in \mathcal{D}_{L}^{+c}, 
        \end{aligned}
    \end{equation}
    Analogous to \eqref{eq:coeff_vector}, this yields the minimization problem:
    \begin{equation}
        \begin{aligned}
            &\min_{\mathbf{a}} \quad \mathbf{v}^{\top}\Tilde{\mathbf{M}}\mathbf{v}, \\
            &\textrm{s.t.} \quad \mathbf{v} = [1 \quad \mathbf{a}^{\top}]^{\top}, \\
            & \Tilde{\mathbf{M}} = \frac{1}{L}\Big(\sum_{\mathbf{x}\in \mathcal{D}_{L}^{-c}}\Tilde{\mathbf{M}}_{-c}(\mathbf{x})+\sum_{\mathbf{x}\in \mathcal{D}_{L}}\Tilde{\mathbf{M}}_{0}(\mathbf{x})+\sum_{\mathbf{x}\in \mathcal{D}_{L}^{+c}}\Tilde{\mathbf{M}}_{+c}(\mathbf{x})\Big),
        \end{aligned}
    \end{equation}
    which is mathematically equivalent to the following quadratic problem with a closed-form solution:
    \begin{equation}
        \begin{aligned}
            \mathbf{a}^{*}&=\argmin_{\mathbf{a}}\mathbf{a}^{\top}\Big(\mathbf{M}_{\mathcal{D}_{L}^{-c}}+\mathbf{M}_{\mathcal{D}_{L}}+\mathbf{M}_{\mathcal{D}_{L}^{+c}}\Big)\mathbf{a} \\ 
            &+c\Big(\sum_{\mathbf{x}\in \mathcal{D}_{L}^{-c}}\mathbf{b}(\mathbf{x})-\sum_{\mathbf{x}\in \mathcal{D}_{L}^{+c}}\mathbf{b}(\mathbf{x})\Big)^{\top}\mathbf{a}.
        \end{aligned}
    \end{equation}
    
    Because only noisy realizations $\mathbf{y} \in \Tilde{\mathcal{D}}_{L}$ are available, we approximate the ribbon datasets using the geometric scaling method \cite{9336312}:
    \begin{equation} \label{eq:geometric_scaling}
        \begin{cases}
          g((1-c)\mathbf{x})=c \\
          g(\mathbf{x})= 0 \\
          g((1+c)\mathbf{x})=-c
        \end{cases}
    \end{equation}
    Using a scaling factor of $c=0.05$, we apply this transformation to the noisy samples to generate the augmented sets $\Tilde{\mathcal{D}}_{L}^{+c}$ (comprising points $(1-c)\mathbf{y}$) and $\Tilde{\mathcal{D}}_{L}^{-c}$ (comprising points $(1+c)\mathbf{y}$). RAIN-FIT is subsequently applied to these sets to compute the matrix estimates $\hat{\mathbf{M}}_{\mathcal{D}_{L}^{-c}}$, $\hat{\mathbf{M}}_{\mathcal{D}_{L}}$, and $\hat{\mathbf{M}}_{\mathcal{D}_{L}^{+c}}$, successfully yielding a smooth fit for the corrupted data.

       \begin{figure}[t]
        \centering        \includegraphics[scale=0.3]{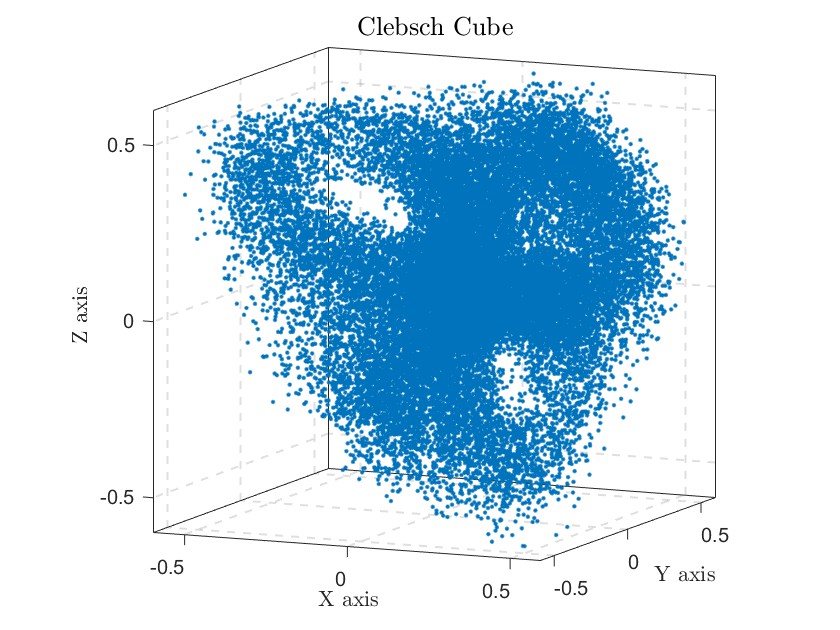}
        \caption{Point cloud of the Clebsch cubic surface corrupted by noise at a 20\% level.}
         \label{fig:ClebNoisy20}
        \end{figure}

\section{Numerical results} \label{sec:numerical_results}
While the proposed approach applies to arbitrary $\mathbb{R}^{n}$ spaces, we focus on $\mathbb{R}^{2}$ and $\mathbb{R}^{3}$ for visualization purposes. For extensive simulation results, including comprehensive benchmarks against baseline methods such as SIREN, Poisson Reconstruction, and Encoder-X, we refer the reader to \cite{sleem2026rain}.

\subsection{Experimental Setup}
\subsubsection{Feature Functions} \label{subsub:basis_desc}
We define the feature functions $\mathbf{b}(\mathbf{x}):\mathbb{R}^{n}\rightarrow \mathbb{R}^{N}$, characterized by the pair $(\gamma, \mathcal{C}_{\gamma})$, utilizing $n$-dimensional monomials of degree up to $\gamma$ as our basis functions, arranged in lexicographical order.

            \begin{figure}[t]
             \centering
             \begin{subfigure}[b]{0.48\textwidth}
                 \centering
                 \includegraphics[scale=0.3]{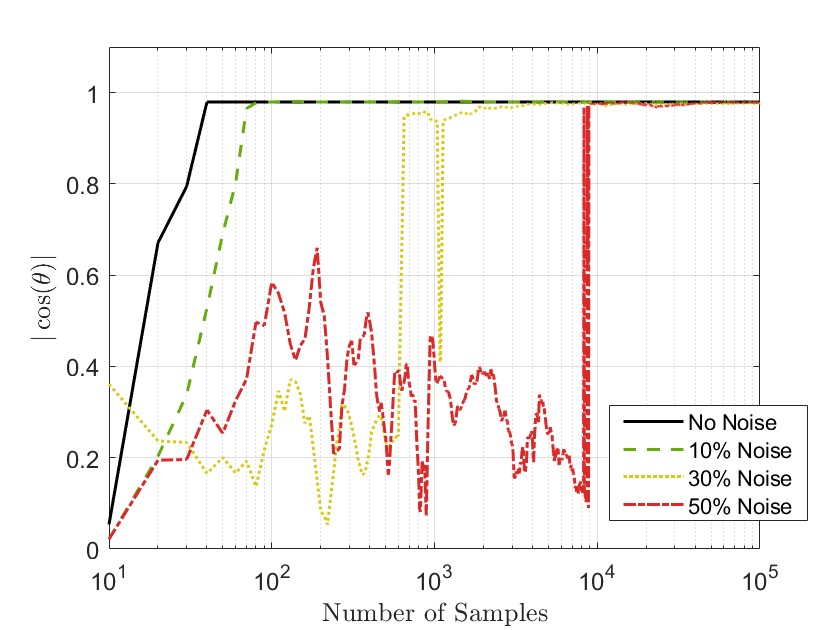}
                 \caption{$|\cos{\theta}|$ vs. the number of data samples for different noise levels.}
                 \label{fig:clebsch}
             \end{subfigure}
             \hspace{0.15cm}
             \begin{subfigure}[b]{0.48\textwidth}
                 \centering
                 \includegraphics[scale=0.3]{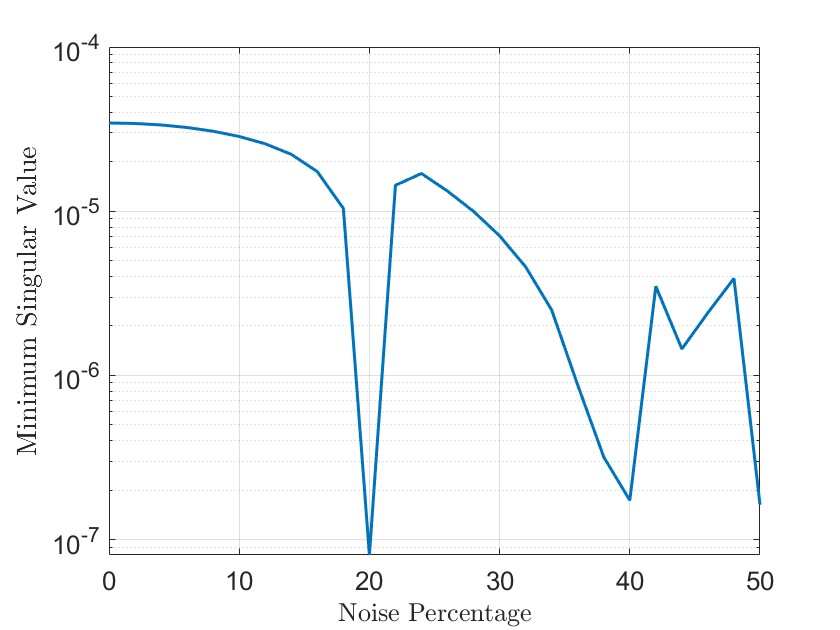}
                 \caption{Minimum Singular value of $\Hat{\mathbf{M}}_{\mathcal{D}_{L}}$ vs. noise percentage.}
                 \label{fig:search_clebsch}
             \end{subfigure}
            \caption{Convergence and parameter estimation for the Clebsch cubic. (a) Absolute cosine similarity $|\cos\theta|$ between the estimated and optimal coefficient vectors, demonstrating asymptotic recovery across varying noise levels. (b) Grid search over the noise parameter; the global minimum singular value of $\hat{\mathbf{M}}_{\mathcal{D}_{L}}$ accurately identifies the true injected noise level ($20\%$).}
            \label{fig:clebsch_results}
            \end{figure}

\subsubsection{Added Noise} \label{subsub:noise_desc}
We assume zero-mean IID uniform noise $\epsilon_{i}\sim U[-u,u]$ for $i\in [n]$ is added to all data points. The unknown noise parameter vector $\boldsymbol\theta\in\mathbb{R}$ thus consists solely of the scalar noise bound $u$. This magnitude $u$ is dataset-dependent, defined as $\max_{\mathbf{x} \in \mathcal{D}_{L}} \|\mathbf{x}\|_{\infty}$. For example, a 10\% noise level applied to an arbitrary dataset $\mathcal{D}_{L}=\{[-0.8, -0.2]^{\top}, [0.5, 0.2]^{\top}\}$ yields $u = 0.1 \times 0.8 = 0.08$.

            \begin{figure}[t]
    \centering
    \includegraphics[scale=0.23]{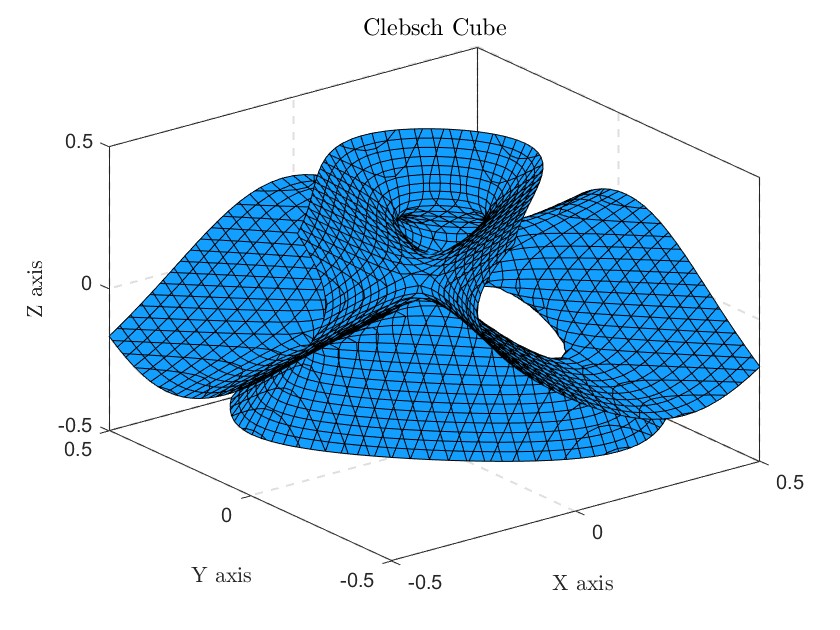}
    \caption{Zero-level set reconstruction of the Clebsch cubic utilizing the proposed RAIN-FIT algorithm on the $20\%$ corrupted dataset. The global topology is accurately recovered without relying on local normal estimates.}
    \label{fig:Cleb20}
\end{figure}

\subsection{3D Surface Fitting}
We present the efficacy of RAIN-FIT in fitting surfaces that can be represented by the zero set of a function, specifically, surfaces satisfying Assumption \ref{A1}. To quantify the accuracy of our surface recovery, we measure the angular alignment between the estimated coefficient vector $\hat{\mathbf{a}}_{L}$ and the analytically known optimal vector $\mathbf{a}^{*}$ using the absolute cosine similarity:
\begin{equation} \label{eq:cos_sim}
    |\cos{\theta}|=\frac{|\hat{\mathbf{a}}_{L}^{\top}\mathbf{a}^{*}|}{\|\hat{\mathbf{a}}_{L}\|\|\mathbf{a}^{*}\|}.
\end{equation}            
The absolute value is employed because the vectors $\mathbf{a}^{*}$ and $-\mathbf{a}^{*}$ define identical zero-level sets, rendering their sign difference geometrically inconsequential. 

Furthermore, to address realistic scenarios outlined in Assumption \ref{A1}—where the noise distribution is known (e.g., uniform) but its precise bounds are not—RAIN-FIT employs a straightforward grid search. We parameterize $\Hat{\mathbf{M}}(., f_{.})$ across varying noise levels and select the parameter $\boldsymbol\theta$ that yields the minimum singular value for the empirical matrix $\Hat{\mathbf{M}}_{\mathcal{D}_{L}}$. This optimally recovered parameter is then utilized to deduce $\hat{\mathbf{a}}_{L}$.

\subsubsection{Clebsch cube}To evaluate 3D performance on a topologically complex surface, we test RAIN-FIT on the Clebsch cubic \cite{clebsch}. The challenge of fitting this geometry is illustrated in Fig.~\ref{fig:ClebNoisy20}, which displays the target point cloud severely corrupted by noise. As a third-degree implicit polynomial, we set $\gamma=3$ and $n=3$. We analytically compute the optimal vector $\mathbf{a}^{*}$ and measure its alignment with $\hat{\mathbf{a}}_{L}$ via the cosine similarity \eqref{eq:cos_sim}.

As shown in Fig.~\ref{fig:clebsch}, RAIN-FIT robustly recovers $\mathbf{a}^{*}$ despite severe noise obfuscating structural details. Higher noise levels and the surface's higher-order polynomial nature demand proportionally larger sample sizes for convergence. The jagged transient response is characteristic of a single stochastic realization rather than an ensemble average.

To identify the unknown noise parameter $\boldsymbol\theta$, Fig.~\ref{fig:search_clebsch} plots the minimum singular value of $\Hat{\mathbf{M}}_{\mathcal{D}_{L}}$ across candidate noise levels for a dataset corrupted by 20\% noise. The minimum perfectly coincides with the true injected noise, validating our grid-search methodology for parameterizing $\Hat{\mathbf{M}}(\cdot, f_{\cdot})$. 

Utilizing this estimated parameter on the highly corrupted point cloud (Fig.~\ref{fig:ClebNoisy20}), RAIN-FIT accurately reconstructs the complex Clebsch surface (Fig.~\ref{fig:Cleb20}) using only a few thousand samples.
To demonstrate the fragility of methods dependent on local geometric estimates, Fig.~\ref{fig:Clebsch20Po} illustrates the output of Poisson surface reconstruction applied to the same dataset corrupted by a 20\% noise level. The stochastic perturbations disrupt the spatial coherence required for continuous normal vector fields, causing the algorithm to generate severe topological distortions and fail completely in recovering the true underlying manifold. Furthermore, SIREN additionally demands sample sizes in the millions to even initiate training. For comprehensive visual comparisons demonstrating the breakdown of these baseline methods, we direct the reader to \cite{sleem2026rain}.

\begin{figure}[t]
    \centering
    \includegraphics[scale=0.07]{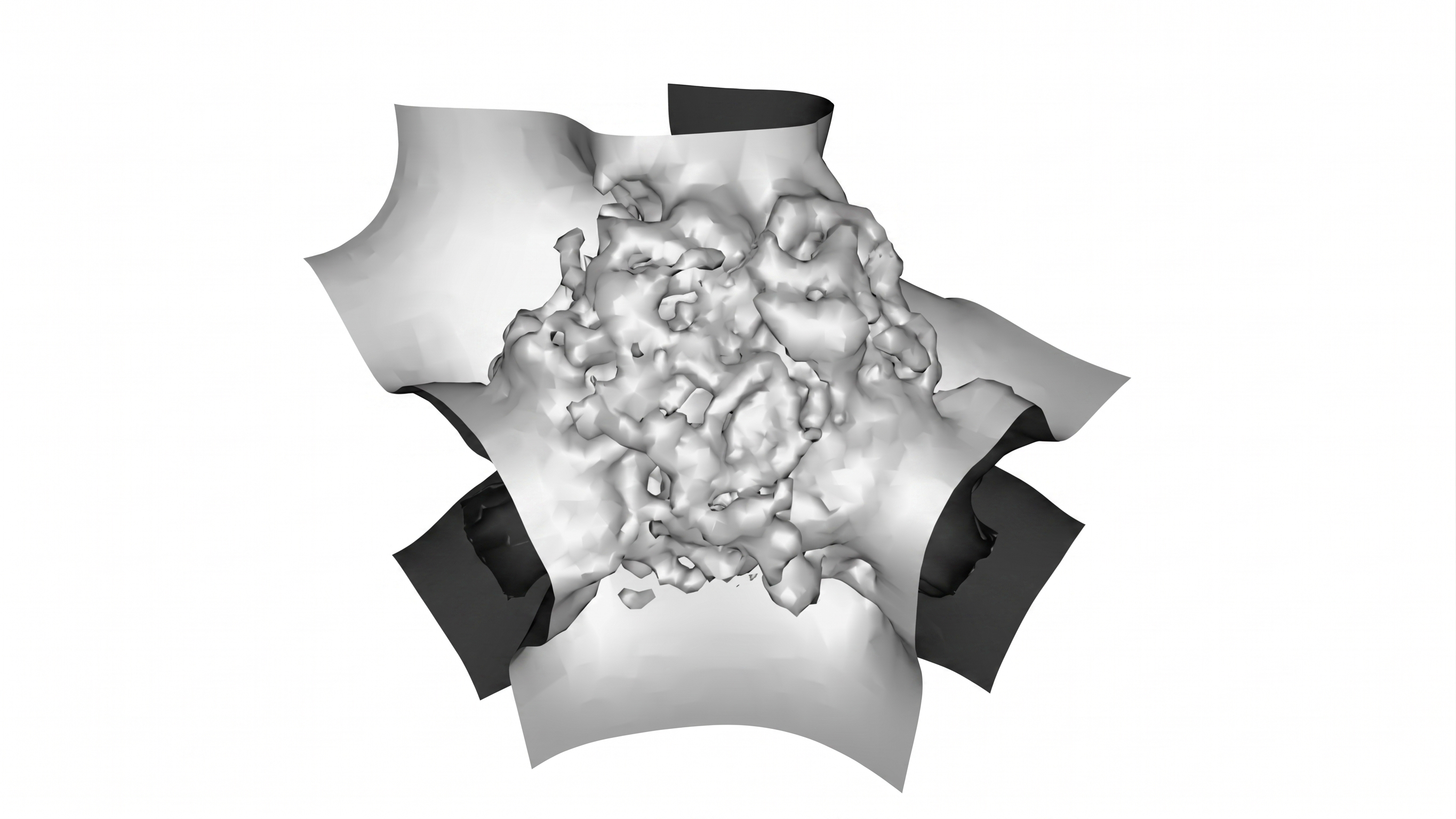}
    \caption{Failure of standard Poisson surface reconstruction on the $20\%$ corrupted Clebsch dataset. The stochastic disruption of the oriented normal vector field results in severe topological distortions.}
    \label{fig:Clebsch20Po}
\end{figure}

    \begin{figure}[t]
    \centering
    \begin{subfigure}[b]{0.2\textwidth}
        \centering
        \includegraphics[scale=0.23]{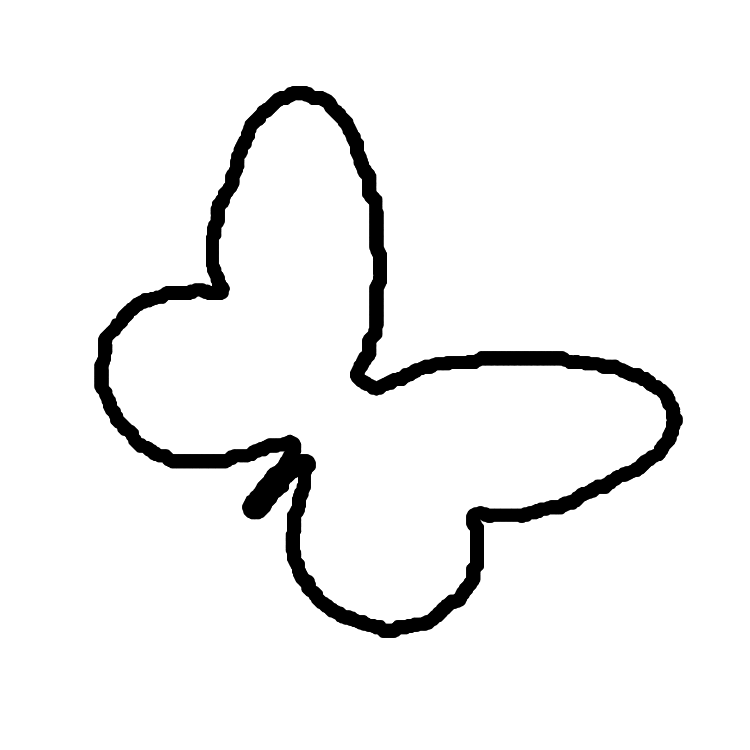}
        \caption*{}
        \label{fig:boot0}
    \end{subfigure}
    \hfill
    \begin{subfigure}[b]{0.2\textwidth}
        \centering
        \includegraphics[scale=0.23]{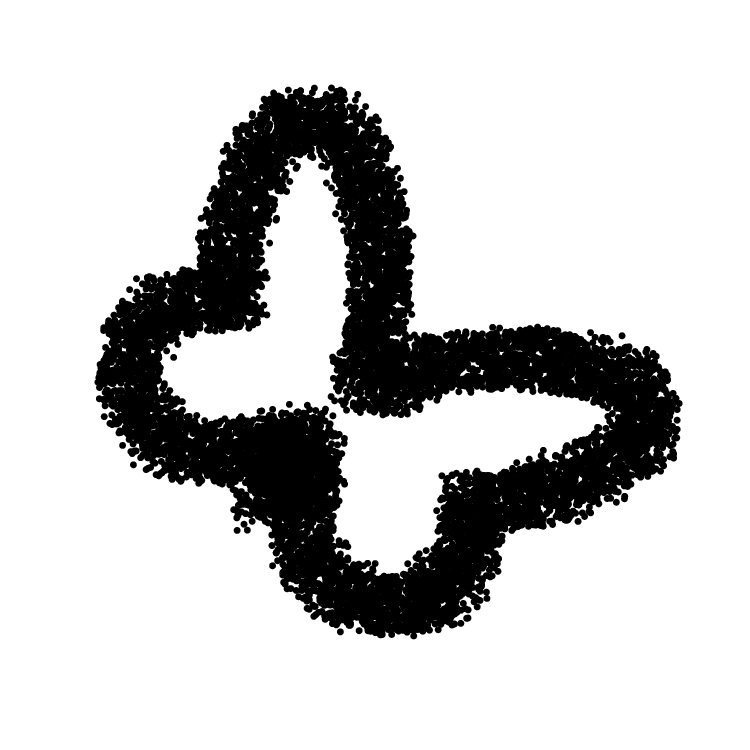}
        \caption*{}
        \label{fig:boot10}
    \end{subfigure}
    \hfill
    \begin{subfigure}[b]{0.2\textwidth}
        \centering
        \includegraphics[scale=0.23]{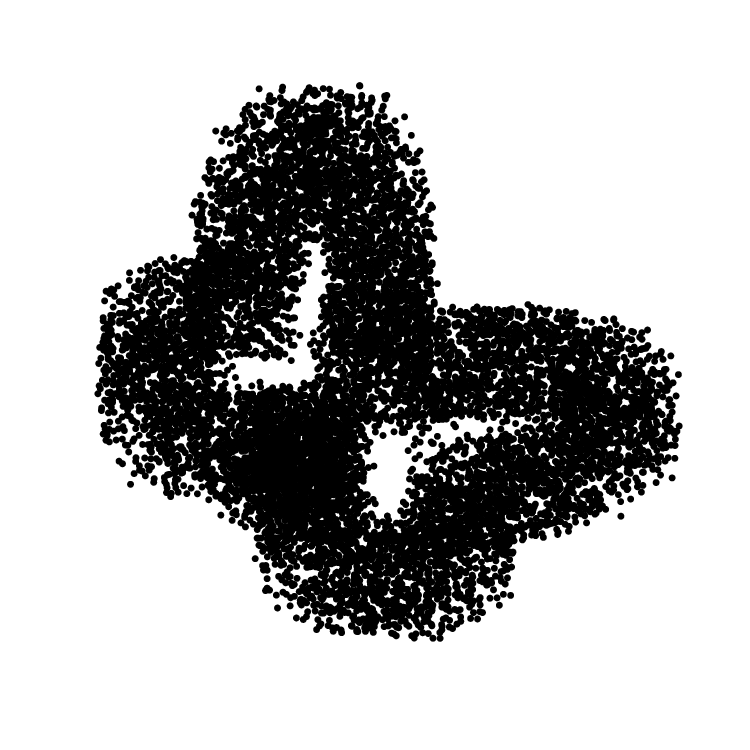}
        \caption*{}
        \label{fig:boot20}
    \end{subfigure}
    \hfill
    \begin{subfigure}[b]{0.2\textwidth}
        \centering
        \includegraphics[scale=0.23]{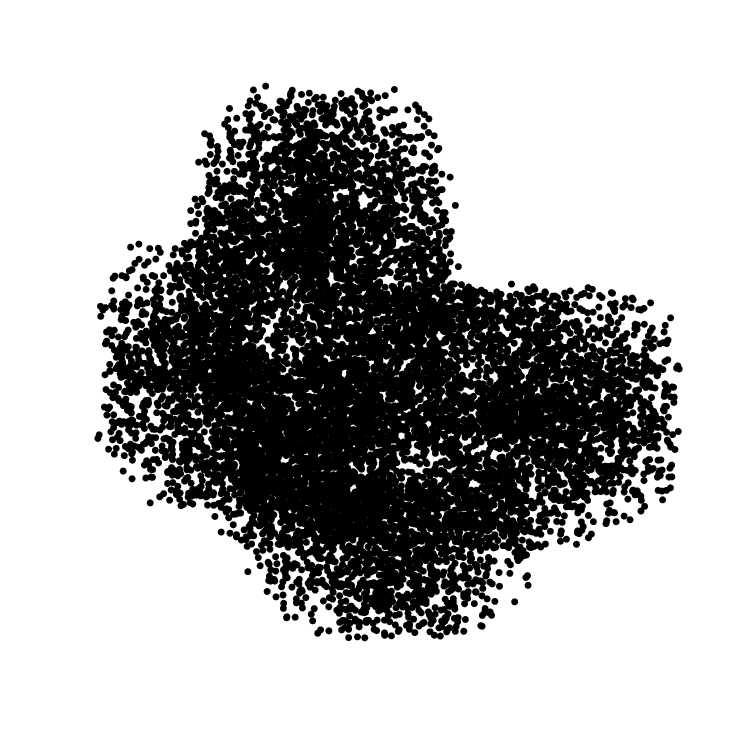}
        \caption*{}
        \label{fig:boot30}
    \end{subfigure}
    \vfill
    \begin{subfigure}[b]{0.2\textwidth}
        \centering
        \includegraphics[scale=0.3, trim={2cm -10cm 0 1cm}]{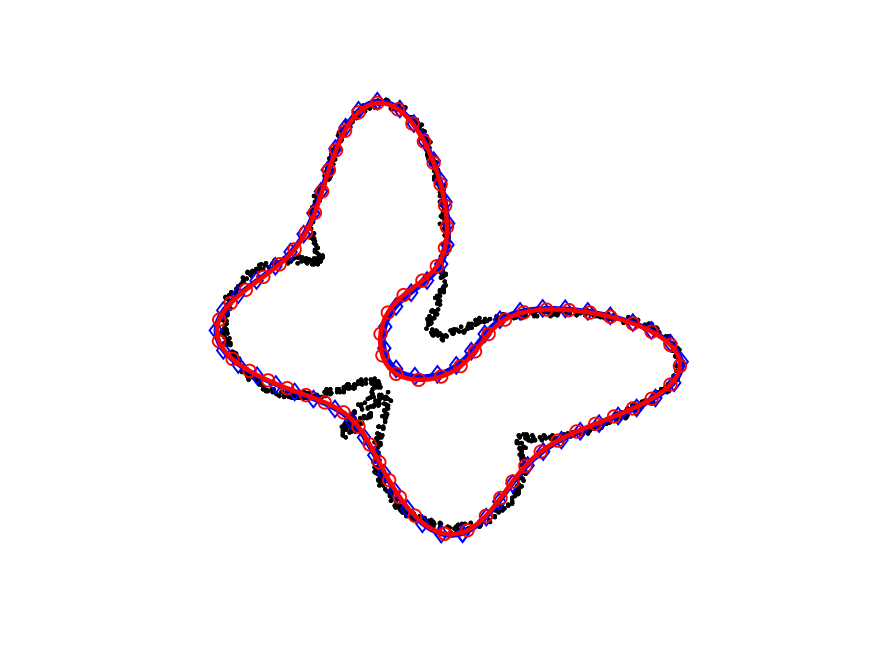}
        \caption*{}
        \label{fig:butterfly0}
    \end{subfigure}
    \hfill
    \begin{subfigure}[b]{0.2\textwidth}
        \centering
        \includegraphics[scale=0.3, trim={0 -10cm 0 1cm}]{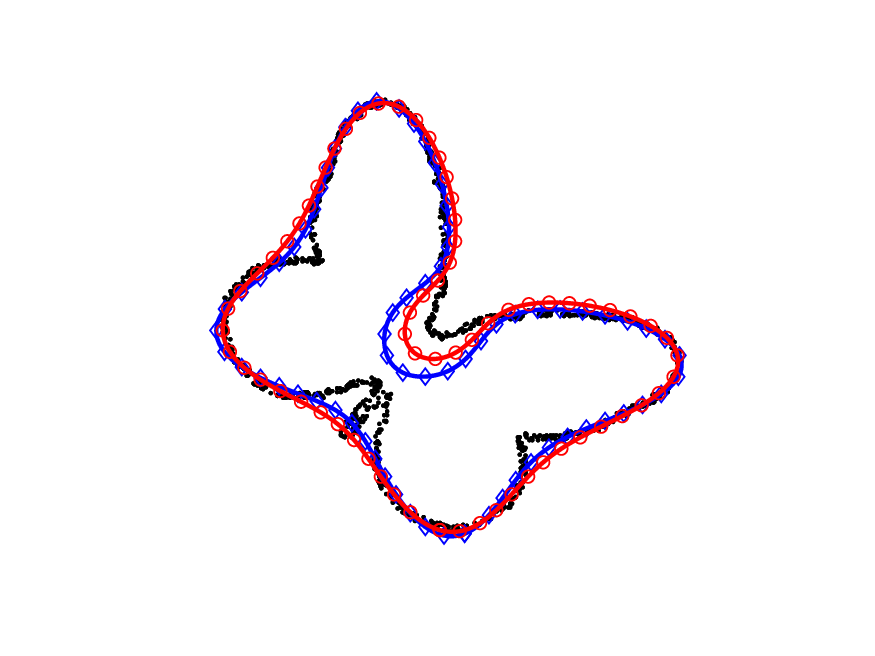}
        \caption*{}
        \label{fig:butterfly10}
    \end{subfigure}
    \hfill
    \begin{subfigure}[b]{0.2\textwidth}
        \centering
        \includegraphics[scale=0.3, trim={0 -10cm 0 1cm}]{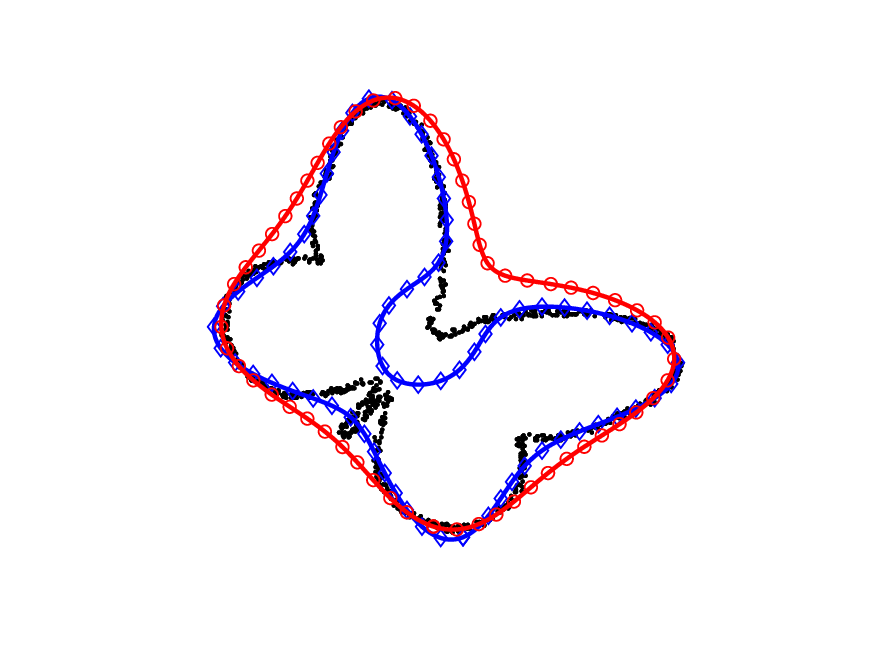}
        \caption*{}
        \label{fig:butterfly20}
    \end{subfigure}
    \hfill
    \begin{subfigure}[b]{0.2\textwidth}
        \centering
        \includegraphics[scale=0.3, trim={0 -10cm 0 1cm}]{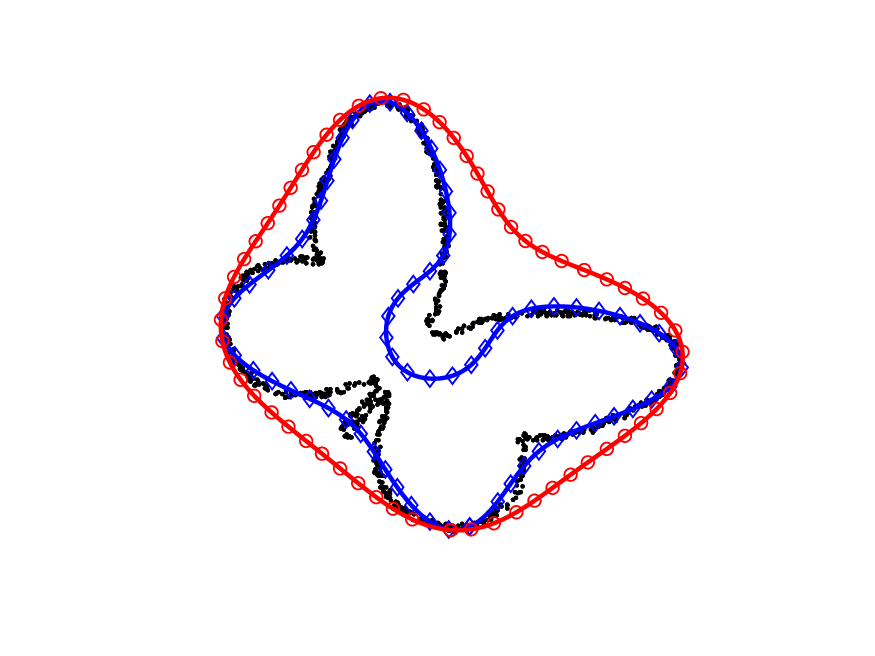}
        \caption*{}
        \label{fig:butterfly30}
    \end{subfigure}
    
    \vspace{-3.5cm}

    \caption{Fitting results employing RAIN-FIT (diamond blue) and Encoder-X (circle red). The black line depicts the original shape to be used as a reference for the quality of the fit. The columns (from left to right) correspond to the results for noise levels of 0\%, 10\%, 20\%, and 30\%.}\label{fig:2D_results}
\end{figure}
\subsection{2D Surface Fitting: Comparison With Encoder-X}
We now investigate $\mathbb{R}^{2}$ shapes that do not conform to zero sets. Data is mean-centered, normalized to a unit maximum Euclidean norm, and scaled via \eqref{eq:geometric_scaling}. Features are configured with $\gamma=4$ and $n=2$.

We benchmark RAIN-FIT against the state-of-the-art Encoder-X neural network \cite{9336312}. Fig.~\ref{fig:2D_results} presents visual comparisons across 0\%, 10\%, 20\%, and 30\% noise levels (the airplane shape is limited to 10\% due to its high aspect ratio). While both methods perform well in low-noise regimes, Encoder-X's reconstruction quality degrades severely at $\geq 20\%$ noise.

To ensure a fair computational evaluation, we modified Encoder-X's stopping criterion to halt training at epoch $k$ when the relative decoder loss \cite{9336312} satisfies $|L(k) - L(k-1)|/|L(k-1)| \leq \delta$. We empirically set $\delta = 10^{-4}$, representing the strictest threshold that guarantees algorithmic termination without degrading the reconstruction quality depicted in Fig. \ref{fig:2D_results}. This ensures we are evaluating Encoder-X at its absolute peak performance capability rather than prematurely halting its convergence. While platform differences preclude a strictly direct comparison---Encoder-X utilized an NVIDIA T4 GPU via PyTorch whereas RAIN-FIT ran on a standard Intel i7 CPU---the temporal advantage remains stark. Empowered by the closed-form formulation in Section \ref{sec:smoothing}, RAIN-FIT maintained a stable runtime of approximately 1.9 seconds across all noise levels. Conversely, Encoder-X required between 32 seconds and several minutes, with computation times degrading heavily under increased data corruption.

\section{Conclusion} \label{sec:conclusion}
In this paper, we introduced RAIN-FIT, a computationally efficient and dimension-agnostic algorithm for robust surface recovery from highly corrupted point clouds. By modeling data as level sets of basis functions, RAIN-FIT explicitly compensates for noise-induced bias, yielding a global mathematical representation of the surface without relying on fragile geometric preprocessing or hyperparameter tuning. We established sufficient theoretical conditions for valid feature functions—demonstrating that standard polynomial, trigonometric, and exponential bases naturally satisfy these criteria—and provided rigorous proofs of almost sure convergence. Furthermore, we extended the framework beyond strict zero-set constraints to accommodate smooth surface fitting in more complex scenarios. Comprehensive numerical evaluations confirm that RAIN-FIT consistently outperforms state-of-the-art baseline methods.

\bibliographystyle{IEEEtran}
\bibliography{Ref}

\end{document}